\documentclass[reqno,12pt]{amsart}

\usepackage{amsmath}
\usepackage{latexsym}
\usepackage{amssymb}
\usepackage{hyperref}
\usepackage{graphicx}
\usepackage{epstopdf}
\usepackage{ifpdf}
\ifpdf
\fi

\makeatletter
 
 \newcommand{\Rmnum}[1]{\expandafter\@slowromancap\romannumeral #1@}
 \makeatother

\newtheorem{theorem}{Theorem}[section]

\newtheorem{proposition}[theorem]{Proposition}

\newtheorem{remark}[theorem]{Remark}

\newcommand{\R}{{\mathbb R}}

\newcommand{\be}{\begin{equation}}
\newcommand{\ee}{\end{equation}}
\newcommand{\bea}{\begin{eqnarray}}
\newcommand{\eea}{\end{eqnarray}}
\newcommand{\ba}{\begin{array}}
\newcommand{\ea}{\end{array}}

\newcommand{\ol}{\overline}

\newcommand{\id}{\mathbb{I}}

\newcommand{\im}{\mathrm{Im}}

\newcommand{\sig}{\sigma}

\newcommand{\lam}{\lambda}

\newcommand{\gam}{\gamma}
\newcommand{\Gam}{\Gamma}

\newcommand{\x}{\xi}

\newcommand{\tha}{\theta}

\numberwithin{equation}{section}

\begin{document}
\title[RHP for the mCH equation]{Long-time asyptotics behavior for the integrable modified Camassa-Holm equation with cubic nonlinearity}

\author[J.Xu]{Jian Xu*}
\address{College of Science\\
University of Shanghai for Science and Technology\\
Shanghai 200093\\
People's  Republic of China}
\email{corresponding author: jianxu@usst.edu.cn}

\author[E.Fan]{Engui Fan}
\address{School of Mathematical Sciences, Institute of Mathematics and Key Laboratory of Mathematics for Nonlinear Science\\
Fudan University\\
Shanghai 200433\\
People's  Republic of China}
\email{faneg@fudan.edu.cn}


\date{\today}

\begin{abstract}
In this paper, we investigate the long-time asymptotic behavior of the solution to the initial value problem for the modified Camassa-Holm (mCH) equation with cubic nonlinearity. The equation is known to be integrable, which we mean it admits an Lax pair. We formulate the initial value problem as an associate vector Riemann-Hilbert problem, which allows us to give a parametric representation of the solution to the initial value problem in terms of the solution of the Riemann-Hilbert problem. And then by adopting the nonlinear steepest descent method, we can get the explicit leading order asymptotic of the solution as time goes to infinity.

\end{abstract}

\maketitle

\section{Introduction}
In this paper, we concern the initial value problem (IVP) for the modified Camassa-Holm (mCH) equation:
\begin{subequations}
\be\label{mche}
m_{t}+(m(u^2-u^2_{x}))_{x}+\kappa u_{x}=0,\quad x\in \R,t>0,\quad m=u-u_{xx}.
\ee
\be\label{intdata}
u(x,0)=u_0(x),\quad x\in\R,
\ee
\end{subequations}
where $u=u(x,t)$ is a real-valued function of spatial variable $x$ and time $t$, and the subscripts $x$
and $t$ appended to $m$ and $u$ denote partial differentiation, and $\kappa$ is a positive constant.
The equation (\ref{mche}) was proposed as a new integrable system by Fuchssteiner \cite{fpd1996} and Olver and Rosenau \cite{orpre1996} by applying the general method of tri-Hamiltonian duality to the bi-Hamiltonian representation
of the modified Korteweg-deVries equation. Later, it was obtained by Qiao \cite{qjmp2006}
from the two-dimensional Euler equations, where the variables $u(x,t)$ and $m(x,t)$ represent,
respectively, the velocity of the fluid and its potential density. In many literatures, for instance \cite{gloqcmp2013,matsuno2014}, the mCH equation was said can be solved by the method of inverse scattering because it admits a Lax pair \cite{qltmp2011}. To the authors knowledge, however, there are no articles to construct the solution $u(x,t)$ of the mCH equation (\ref{mche}) by using inverse scattering transform method. In this paper, we use a similar way as the spectral analysis of the short pulse equation in \cite{jjde2018} to formulate the initial value problem to a Riemann-Hilbert problem, as the Lax pair of the mCH equation (\ref{mche}) is the Wadati-Konno-Ichikawa (WKI)-type, too.
And then, we derive the leading order asymptotic behavior of the solution $u(x,t)$ as $t\rightarrow \infty$ by using the nonlinear steepest descent method.
\par

%
%
%
\par
{\bf Organization of the paper:} In section 2, since the associated Lax pair of mCH equation has singularities at $\lambda=0$ and $\lambda=\infty$, we perform the spectral analysis to deal with the two singularities, respectively. However, we just formulate the associated vector Riemann-Hilbert in an alternative space variable $y$ instead of the original space variable $x$. Hence, we can reconstruct the solution $u(x,t)$ parameterized from the solution of the Riemann-Hilbert problem via the asymptotic behavior of the spectral variable at $\lambda=0$. Fortunately, we can also obtain the asymptotic relation between $y$ and $x$ when analyzing the vector Riemann-Hilbert problem by using the nonlinear steepest descent method. Hence, we can calculate the leading order asymptotic behavior of the solution $u(x,t)$ in section 3. Then, in section 4, we obtain the soliton solutions under the assumption that the spectral data $a(k)$ has finite simple poles.

\section{Riemann-Hilbert Problem}

In this section, we show the solution of the IVP for mCH equation (\ref{mche}) can be constructed in terms of the solution of a Rieamnn-Hilbert (RH) problem.
The mCH equation is an integrable nonlinear partial differential equation which admits the following Lax pair,
\begin{subequations}\label{psilax}
\be
\Psi_x(x,t,\lam)=U(x,t,\lam)\Psi,
\ee
\be
\Psi_t(x,t,\lam)=V(x,t,\lam)\Psi,
\ee
\end{subequations}
where
\begin{subequations}
\be
U(x,t,\lam)=\frac{1}{2}\left(\ba{cc}-Q&\lam m(x,t)\\-\lam m(x,t)&Q\ea\right),
\ee
\be
V(x,t,\lam)=\left(\ba{cc}\frac{Q}{\lam^2}+\frac{1}{2}Q(u^2-u^2_{x})&-\frac{u-Qu_{x}}{\lam}-\frac{1}{2}\lam(u^2-u^2_{x})m\\ \frac{u+Qu_x}{\lam}+\frac{1}{2}\lam(u^2-u^2_{x})m&-\frac{Q}{\lam^2}-\frac{1}{2}Q(u^2-u^2_{x})\ea\right),
\ee
\end{subequations}
with
\be
Q=Q(\lam,\kappa)=\sqrt{1-\frac{1}{2}\kappa \lam^2}.
\ee

If we introduce the following transformations
\be
\left\{
\ba{l}
x=\tilde x,\\
t=\frac{2}{\kappa}\tilde t,\\
u(x,t)=\sqrt{\frac{\kappa}{2}}\tilde u(\tilde x,\tilde t),
\ea
\right.
\ee
then, the mCH equation (\ref{mche}) becomes
\[
\tilde m_{\tilde t}+(\tilde m(\tilde u^2-\tilde u^2_{\tilde x}))_{\tilde x}+2\tilde u_{\tilde x}=0.
\]

Hence, without loss of generally, we can assume that $\kappa=2$ in the following. And we assume that the initial value $u_0(x)$ lies in Schwarz space.
\par
Then,
\be
U(x,t,\lam)=\frac{1}{2}\left(\ba{cc}-Q&\lam m(x,t)\\-\lam m(x,t)&Q\ea\right),\quad Q=Q(\lam,\kappa=2)=\sqrt{1-\lam^2}
\ee

We know that there are two singularity points at $\lam=0$ and $\lam=\infty$. Hence, we need two different transformations to control the behavior of the eigenfunctions at these two points, respectively.

\par
To over come the multi-value of the square root and avoid introducing the Riemann surface, we introduce an new spectral variable via the following transformation,
\be\label{Qklamk}
Q(k)=\frac{i}{2}(k-\frac{1}{k}),\quad \lam(k)=\frac{1}{2}(k+\frac{1}{k}).
\ee

\subsection{Spectral analysis at $\lam= \infty$}

Denote the Pauli matrices as
\be
\ba{lll}
\sig_1=\left(\ba{cc}0&1\\1&0\ea\right),&\sig_2=\left(\ba{cc}0&-i\\i&0\ea\right),&\sig_3=\left(\ba{cc}1&0\\0&-1\ea\right)
\ea
\ee

It is obviously known from (\ref{Qklamk}) that $k\rightarrow \infty$ and $k\rightarrow 0$ correspond to $\lam\rightarrow \infty$. So, we need control the behavior as $k\rightarrow \infty$ and $k\rightarrow 0$, respectively, to control the behavior of $\lam\rightarrow \infty$.
\par

{\bf Firstly, we consider the case as $k\rightarrow \infty$.}

Let
\be\label{Gdef}
G(x,t)=\sqrt{\frac{\sqrt{m^2+1}+1}{2\sqrt{m^2+1}}}\left(\ba{cc}1&\frac{-im}{\sqrt{m^2+1}+1}\\\frac{-im}{\sqrt{m^2+1}+1}&1\ea\right),
\ee
and
\be
p(x,t,k)=x-\int_{x}^{\infty}(\sqrt{m^2(x',t)+1}-1)dx'-\frac{2t}{\lam(k)^2}.
\ee


Making a transformation as
\be\label{psirelmu}
\Psi(x,t,k)=G(x,t)\mu(x,t,k)e^{-\frac{Q(k)}{2} p(x,t,k)\sig_3},
\ee
then
\be\label{mulaxpair}
\left\{
\ba{l}
\mu_x+\frac{Q(k)}{2}p_x[\sig_3,\mu]=\tilde{U}(x,t,k)\mu,\\
\mu_t+\frac{Q(k)}{2}p_t[\sig_3,\mu]=\tilde{V}(x,t,k)\mu,
\ea
\right.
\ee
where
\begin{subequations}
\be
\tilde{U}(x,t,k)=\frac{im_x}{2(m^2+1)}\left(\ba{cc}0&1\\1&0\ea\right)+\frac{1}{2k}\frac{m}{\sqrt{m^2+1}}\left(\ba{cc}-im&1\\-1&im\ea\right),
\ee
\be
\ba{l}
\tilde{V}(x,t,k)=\frac{im_t}{2(m^2+1)}\left(\ba{cc}0&1\\1&0\ea\right)-\frac{m(u^2-u^2_{x})}{2k\sqrt{m^2+1}}\left(\ba{cc}-im&1\\-1&im\ea\right)+\frac{(k^2-1)u_x}{k^2+1}\left(\ba{cc}0&1\\1&0\ea\right)\\
{}-\frac{2ku}{(k^2+1)\sqrt{m^2+1}}\left(\ba{cc}-im&1\\-1&im\ea\right)+\frac{2ik(k^2-1)}{(k^2+1)^2}\left(\ba{cc}\frac{1}{\sqrt{m^2+1}}-1&\frac{-im}{\sqrt{m^2+1}}\\\frac{im}{\sqrt{m^2+1}}&1-\frac{1}{\sqrt{m^2+1}}\ea\right).
\ea
\ee
\end{subequations}

\subsubsection{Eigenfunctions $\mu_j(x,t,k)$}

Define two eigenfunctions $\mu_{j}(x,t,k),j=1,2$,
\begin{subequations}\label{mujdef}
\be\label{mu1def}
\mu_{1}(x,t,k)=\id+\int_{-\infty}^{x}e^{-\frac{Q(k)}{2}(p(x,t,k)-p(y,t,k))\hat\sig_3}(\tilde{U}(y,t,k)\mu_{1}(y,t,k))dy,
\ee
\be\label{mu2def}
\mu_{2}(x,t,k)=\id+\int_{\infty}^{x}e^{-\frac{Q(k)}{2}(p(x,t,k)-p(y,t,k))\hat\sig_3}(\tilde{U}(y,t,k)\mu_{2}(y,t,k))dy,
\ee
\end{subequations}

\begin{proposition}[Analytic property]\label{Anlyper}
Then, $\{\mu_{j}(x,t,k)\}_{j=1}^{2}$ satisfy the following the bounded and analytic properties,
\be
\left\{
\ba{l}
\mu_{1}\in(D_1,D_2)\\
\mu_{2}\in(D_2,D_1).
\ea
\right.
\ee
Here, $D_1$ and $D_2$ denote the upper and lower half-plane, respectively.
\end{proposition}

\par

\begin{proposition}[Symmetry peoperty]
The functions $\mu_j(x,t,k)$ have the following symmetry conditions:
\be\label{symrel}
\sig_{1}\mu_{\pm}(x,t,- k)\sig_1=\mu_{\pm}(x,t,k),\quad
\sig_{2}\ol{\mu_{\pm}(x,t,\bar k)}\sig_2=\mu_{\pm}(x,t,k).
\ee
\end{proposition}
\begin{proof}
From the definition (\ref{mujdef}) of the functions $\mu_j(x,t,k)$, a direct computation shows (\ref{symrel}) holds.
\end{proof}

\subsubsection{Scattering matrix}
The eigenfunctions $\mu_1(x,t,k)$ and $\mu_2(x,t,k)$, which are not independent, satisfy the relation for some matrix $s(k)$ independent of $(x,t)$,
\be\label{Sk}
\mu_{-}(x,t,k)=\mu_{+}(x,t,k)e^{-\frac{Q(k)}{2}p(x,t,k)\hat\sig_3}s(k),
\ee
where
\be\label{Skdef}
s(k)=\left(\ba{cc}a(k)&-\ol{b(\bar k)}\\b(k)&\ol{a(\bar k)}\ea\right).
\ee
Then, from (\ref{Sk}), we have
\be
a(k)=\det{([\mu_{-}]_1,[\mu_{+}]_2)}.
\ee
By the analytic property (\ref{Anlyper}), we know that $a(k)$ is analytic in $D_1$.
\par

\subsubsection{Spectral analysis at $k=0$}
We know that $\lam(k)$ remains the same if we use $\frac{1}{k}$ to replace $k$, from the definition (\ref{Qklamk}) of $\lam(k)$.  Hence, we can use this fact to analyse the behavior at $k=0$.
\par
From the Lax pair of $\Psi(x,t,\lam(k))$ (\ref{psilax}), we know that
\be\label{psi1krelpsik}
\Psi(x,t\frac{1}{k})=\sig_2\Psi(x,t,k)\sig_2.
\ee

Then, by the transformation (\ref{psirelmu}) it implies that
\be\label{mu1krelmuk}
\mu_(x,t,\frac{1}{k})=(G^{-1}(x,t))^{2}\sig_2\mu(x,t,k)\sig_2,
\ee
here we use the relation $\sig_2 G(x,t) \sig_2=(G(x,t))^{-1}$ in view of the definition of $G(x,t)$ (\ref{Gdef}).
\par
This equation (\ref{mu1krelmuk}), the relation between $\mu(x,t,k)$ and $S(k)$ (\ref{Sk}) and the symmetry conditions (\ref{sym1}) imply the following proposition.
\begin{proposition}
The functions $\mu(x,t,k)$ and $s(k)$ satisfy the following symmetry conditions:
\begin{subequations}
\be
\mu(x,t,-\frac{1}{k})=(G^{-1}(x,t))^{2}\sig_3\mu(x,t,k)\sig_3,\quad \ol{\mu(x,t,\frac{1}{\bar k})}=(G(x,t))^2\mu(x,t,k),
\ee
\be\label{s1krelsk}
s(-\frac{1}{k})=\sig_3 s(k)\sig_3,\quad \ol{s(\frac{1}{\bar k})}=s(k).
\ee
\end{subequations}
\end{proposition}

\subsection{pre-Riemann-Hilbert Problem}
Define
\be
M(x,t,k)=\left\{
\ba{ll}
(\frac{[\mu_{-}]_1}{a(k)},[\mu_{+}]_2),&\im k>0,\\
([\mu_{+}]_1,\frac{[\mu_{-}]_2}{\ol{a(\bar k)}}),&\im k<0.
\ea
\right.
\ee

\par
Then, we can show $M(x,t,k)$ satisfies the Riemann-Hilbert problem:
\begin{itemize}\label{Mxtrhp}
\item Jump condition:
\be
M_+(x,t,k)=M_-(x,t,k)\left(\ba{cc}1+|r(k)|^2&\ol{r(k)}e^{-Q(k) p(x,t,\lam)}\\r(k)e^{Q(k) p(x,t,\lam)}&1\ea\right),\quad k\in \R,
\ee
where
\be
r(k)=\frac{b(k)}{a(k)}.
\ee
\item Normalize condition˙
\be\label{Mxkinfasy}
M(x,t,k)=\id+\frac{D^{(1)}(x,t)}{k}+O(k^{-2}),\quad k\rightarrow \infty,
\ee
where the off-diagonal entries of the matrix $D^{(1)}(x,t)$ are
\be
D^{(1)}_{12}=\frac{m_x}{(1+m^2)^{\frac{3}{2}}},\quad D^{(1)}_{21}=-\frac{m_x}{(1+m^2)^{\frac{3}{2}}}.
\ee
\end{itemize}
\begin{remark}\label{mxkasy}
The asymptotic expansion formula (\ref{Mxkinfasy}) of $M(x,t,k)$ as $k\rightarrow \infty$ can be derived by substituting the following expansion
\be\label{Mxkasyexp}
M(x,t,k)=D^{(0)}(x,t)+\frac{D^{(1)}(x,t)}{k}+O(k^{-2}),\quad k\rightarrow \infty,
\ee
into the equations (\ref{mulaxpair}), and comparing the order of $k$.
\end{remark}

It shows that it has some difficulties if we want to construct the solution $u(x,t)$ of the mCH equation (\ref{mche}) in terms of the solution of the Riemann-Hilbert problem of $M(x,t,k)$, i.e, (\ref{Mxtrhp}), as $k\rightarrow \infty$, in view of the asymptotic behavior (\ref{Mxkinfasy}) of $M(x,t,k)$. Hence, we should choose another expansion formula of $M(x,t,k)$ to contain the information of the solution $u(x,t)$. Because the Lax pair (\ref{mulaxpair}) have other singularities at $k=\pm i$ corresponding to $\lam=0$, we can control the behavior of the eigenfunctions $\mu_j(x,t,k)$ at $k=\pm i$ to get the information of $u(x,t)$. 

\subsection{Spectral analysis at ($\lam=0$)}
Define another transformation as follows,
\be
\Psi=\mu^{0}e^{(-\frac{Q(k)}{2}x+\frac{Q(k)}{\lam^2(k)}t)\sig_3}
\ee
Then, the Lax pair of $\mu^{0}$ is
\be
\left\{
\ba{l}
\mu_x^{0}+\frac{Q(k)}{2}[\sig_3,\mu^{0}]=U^{0}\mu^{0},\\
\mu_t^{0}-\frac{Q(k)}{\lam^2(k)}[\sig_3,\mu^{0}]=V^{0}\mu^{0},
\ea
\right.
\ee
where
\begin{subequations}
\be
U^{0}=\frac{\lam(k)}{2}m\left(\ba{cc}0&1\\-1&0\ea\right),
\ee
\be
V^{0}=\frac{1}{2}(u^2-u_x^2)\left(\ba{cc}Q(k)&-\lam(k)m\\\lam(k)m&-Q(k)\ea\right)+\frac{u}{\lam(k)}\left(\ba{cc}0&-1\\1&0\ea\right)+\frac{Q(k)}{\lam(k)}u_x\sig_1
\ee
\end{subequations}

As the definitions of $\mu_j(x,t,k)$ (\ref{mujdef}), we can define two eigenfunctions $\{\mu^{(0)}_{j}(x,t,\lam)\}_{j=1}^{2}$. 

Then, a similar computation as (\ref{mxkasy}) shows that the asymptotic behavior of the eigenfunctions $\mu^{0}(x,t,k)$ as $k\rightarrow i$,
\be\label{mu0asyk}
\mu^{0}=\id+\left(\ba{cc}0&-\frac{1}{2}(u+u_x)\\-\frac{1}{2}(u-u_x)&0\ea\right)(k-i)+O\left((k-i)^2\right).
\ee

\subsection{The relation between $\mu(x,t,k)$ and $\mu^{0}(x,t,k)$}
Notice that $\mu_{j}(x,t,k)$ and $\mu^{0}_{j}(x,t,k)$ are the solutions to the same equation about $\Psi(x,t,k)$, then they satisfy the following relation,
\be
\left\{
\ba{l}
\mu_{1}(x,t,k)=G^{-1}(x,t)\mu^{0}_{1}(x,t,k)e^{\frac{Q(k)}{2}\int_{-\infty}^{x}(\sqrt{m^2(x',t)+1})dx'\sig_3},\\
\mu_{2}(x,t,k)=G^{-1}(x,t)\mu^{0}_{2}(x,t,k)e^{-\frac{Q(k)}{2}\int_{x}^{+\infty}(\sqrt{m^2(x',t)+1})dx'\sig_3}.
\ea
\right.
\ee

Hence, from the definition of $a(k)=\det{([\mu_{-}]_1,[\mu_{+}]_2)}$,
we have
\be
a(k)=e^{-\frac{1}{2}\int_{-\infty}^{+\infty}(\sqrt{m^2+1}-1)dx}(1+O((k-i)^2)),\quad k\rightarrow i.
\ee
So, we have the asymptotic behavior of $M(x,t,k)$ as $k\rightarrow i$,
\be
M(x,t,k)=G^{-1}\left[\id+\left(\ba{cc}0&-\frac{1}{2}(u+u_x)\\-\frac{1}{2}(u-u_x)&0\ea\right)(k-i)+O\left((k-i)^2\right)\right]e^{\frac{1}{2}c_+\sig_3}
\ee
where
\be
c_+=\int_{x}^{+\infty}(\sqrt{m^2(x',t)+1}-1)dx'.
\ee

%
%
%
%
%

So, define
\be
\tilde M(x,t,k)=\left(\ba{cc}1&1\ea\right)M(x,t,k),
\ee
then, as $k\rightarrow i$, we have
\be
\ba{l}
\tilde M_1(x,t,k)=a(1+b)\left[1-\frac{1}{2}(u-u_x)(k-i)+O((k-i)^2)\right]e^{\frac{1}{2}c_+},\\
\tilde M_2(x,t,k)=a(1+b)\left[1-\frac{1}{2}(u+u_x)(k-i)+O((k-i)^2)\right]e^{-\frac{1}{2}c_+}
\ea
\ee
where
\be
a=\sqrt{\frac{\sqrt{m^2+1}+1}{2\sqrt{m^2+1}}},\quad b=\frac{im}{\sqrt{m^2+1}+1}.
\ee

Then,
\be
\ba{l}
\frac{\tilde M_2(x,t,i)}{\tilde M_1(x,t,i)}=e^{-c_+},\\
\tilde M_1(x,t,k)\times \tilde M_2(x,t,k)=a^2(1+b)^2\left[1-u(k-i)+O((k-i)^2)\right],\\
\tilde M_1(x,t,i)\times \tilde M_2(x,t,i)=a^2(1+b)^2
\ea
\ee
Hence, we have
\be
u(x,t)=\lim_{k\rightarrow i}\frac{1}{k-i}\left(1-\frac{\tilde M_1(x,t,k)\times \tilde M_2(x,t,k)}{\tilde M_1(x,t,i)\times \tilde M_2(x,t,i)}\right)
\ee
and
\be
c_+=-\ln{\left(\frac{\tilde M_2(x,t,i)}{\tilde M_1(x,t,i)}\right)}.
\ee

\subsection{Riemann-Hilbert problem}
The Riemann-Hilbert problem for $M(x,t,k)$ cannot
be used immediately for recovering the solution of mCH equation (\ref{mche}). Since, in the representation of the jump matrix
$e^{-\frac{Q(k)}{2}p(x,t,k)\hat \sig_3}J_0(k)$ the factor $J_0(k)$ is indeed given in terms of the known initial
data $u_0(x)$ but $p(x,t,k)$ is not, it involves $m(x,t)$ which is unknown (and, in fact, is
to be reconstructed).
\par
To overcome this, we introduce the new (time-dependent) scale
\be\label{ydef}
y(x,t)=x-\int_{x}^{+\infty}(\sqrt{m^2(x',t)+1}-1)dx'=x-c_+(x,t).
\ee
in terms of which the jump matrix becomes explicit.
The price to pay for this, however, is that the solution of the initial problem
can be given only implicitly, or parametrically: it will be given in terms of
functions in the new scale, whereas the original scale will also be given in terms
of functions in the new scale.

Define
\be
M(y,t,k)=M(x(y,t),t,k),\quad \tilde M(y,t,k)=\tilde M(x(y,t),t,k).
\ee

Then, we can get the Riemann-Hilbert problem for the new variable $(y,t)$,
\be\label{RHPmy}
M_+(y,t,k)=M_-(y,t,k)e^{-\frac{Q(k)}{2}(y-\frac{2}{\lam^2(k)}t)\hat\sig_3}\left(\ba{cc}1+|r(k)|^2&\ol{r(k)}\\r(k)&1\ea\right),\quad k\in \R,
\ee
where
\be
r(k)=\frac{b(k)}{a(k)}.
\ee
And from the equations (\ref{mulaxpair}), we can find that
\be
M(y,t,k)=\id+O(k^{-1}),\quad k\rightarrow \infty.
\ee
Then, the solution can be obtained as follows,
\begin{subequations}\label{usol}
\be\label{usolyt}
u(y,t)=\lim_{k\rightarrow i}\frac{1}{k-i}\left(1-\frac{\tilde M_1(y,t,k)\times \tilde M_2(y,t,k)}{\tilde M_1(y,t,i)\times \tilde M_2(y,t,i)}\right)
\ee
and
\be\label{usolytox}
x=y+c_+
\ee
with
\be\label{usolc}
c_+=-\ln{\left(\frac{\tilde M_2(y,t,i)}{\tilde M_1(y,t,i)}\right)}.
\ee
\end{subequations}

\section{Long-time asymptotic}

In this section, we use the nonlinear steepest descent method to derive the asymptotic behavior of the solution $u(x,t)$ of the mCH (\ref{mche}) as time $t$ goes to positive infinity. And we assume that there is no zero point of $a(k)$ such that it will make our analysis easily in technical.
\par
Let us recall the Riemann-Hilbert problem of $M(y,t,k)$ (\ref{RHPmy}) obtained in above section.
\be
\left\{
\ba{ll}
M_+(y,t,k)=M_-(y,t,k)J(y,t,k),&k\in \R,\\
M(y,t,k)\rightarrow \id,&k\rightarrow \infty,
\ea
\right.
\ee
where
\be
J(y,t,k)=e^{-\frac{i}{4}(k-\frac{1}{k})[y-\frac{2t}{\frac{1}{4}(k+\frac{1}{k})^2}]\hat \sig_3}J_0(k),
\ee
with
\be
J_0(k)=\left(\ba{cc}1+|r(k)|^2&\ol{r(k)}\\r(k)&1\ea\right).
\ee
Define
\be
\tha(y,t,k)=-\frac{1}{4}(k-\frac{1}{k})[\frac{y}{t}-\frac{2}{\frac{1}{4}(k+\frac{1}{k})^2}],
\ee
and
\be
\tilde \xi=\frac{y}{t},\quad \tilde k(k)=-\frac{1}{4}(k-\frac{1}{k}).
\ee
Then,
\be
\tha(y,t,k)=\tilde k(k)[\tilde \x-\frac{2}{1+4\tilde k(k)^2}].
\ee
Denote $\tha(y,t,k)$ by $\tha(\tilde \x,\tilde k(k))$. Hence, the sign of the function $\tha(y,t,k)$ can be obtained in terms of $\tha(\tilde \x,\tilde k(k))$ from $\tilde k(k)$ to $k$.
\par
Now, let us consider the critical point of $\tha(\tilde \x,\tilde k(k))$,
\be
\frac{d\tha}{d\tilde k}=\tilde \x-\frac{2(1-4\tilde k^2)}{(1+4\tilde k^2)^2}.
\ee
Letting $\frac{d\tha}{d\tilde k}=0$ and denote $\tilde k^2=s$, then, we have
\be
16\tilde \x s^2+8(\tilde \x+1)s+(\tilde \x-2)=0.
\ee
This is a quadratic algebra equation of $s$, it is easily to see that there are four different cases of the sign table of $\tha(y,t,k)$.
\begin{itemize}
  \item Case 1: $\xi>2$. In this case, the solution $u(x,t)$ of the mCH equation is decaying fast as $t\rightarrow \infty$.
  
  \item Case 2: $0<\xi<2$. In this case, the asymptotic behavior of the solution $u(x,t)$ is as follows,
\be\label{uxtasycase2}
u(x,t)=\frac{\sqrt{-2\gam_0}}{\sqrt{\tilde k_1 t(3-4\tilde k_1^2)}}\cos{(\phi_0)}+O\left(\frac{\ln{(t)}}{t}\right),
\ee
where
\be
\ba{l}
\tilde k_1=\sqrt{\frac{\sqrt{1+4\xi}-(1+\xi)}{4\xi}},\quad k_0=-2\tilde k_1+\sqrt{1+4\tilde k^2_1},\\
\gam(k)=-\frac{1}{2\pi}\ln{(1+|r(k)|^2)},\quad \gam_0=\gam(k_0),\\
\phi_0=\frac{\pi}{4}+{\mbox arg}(r(k_0))+{\mbox arg}(\Gam(i\gam_0))+2\gam_0\ln{\frac{(1+4\tilde k^2_1)^{3/2}}{128\tilde k^3_1t(3-4\tilde k^2_1)^{1/2}}}\\
{}\quad -\frac{1}{\pi}\int_{k_0}^{\frac{1}{k_0}}\ln\left(\frac{1+|r(s)|^2}{1+|r(k_0)|^2}\right)\frac{-2k_0}{s^2-k^2_0}ds-2\tilde k_1(x-c_+-\frac{2t}{1+4\tilde k^2_0}),\\
c_+=-\frac{2}{\pi}\int_{k_0}^{\frac{1}{k_0}}\frac{\ln(1+|r(s)|^2)}{1+s^2}ds.
\ea
\ee
    
  \item Case 3: $-\frac{1}{4}<\xi<0$. In this case,  the asymptotic behavior of the solution $u(x,t)$ is as follows,

                \be\label{uxtasycase3}
u(x,t)=\frac{\sqrt{-2\gam_1}}{\sqrt{\tilde k_1 t(3-4\tilde k_1^2)}}\cos{(\phi_1)}-\frac{\sqrt{-2\gam_2}}{\sqrt{\tilde k_2 t(3-4\tilde k_2^2)}}\cos{(\phi_2)}+O\left(\frac{\ln{(t)}}{t}\right),
\ee 
where
\be
\ba{l}
\tilde k_2=\sqrt{\frac{-\sqrt{1+4\xi}-(1+\xi)}{4\xi}},\quad k_1=k_0, \quad k_2=-2\tilde k_2+\sqrt{1+4\tilde k^2_2},\\
\gam_1=\gam(k_1),\quad \gam_2=\gam(k_2),\\
\phi_1=\frac{\pi}{4}+{\mbox arg}(r(k_1))+{\mbox arg}(\Gam(i\gam_1))+2\gam_1\ln{\frac{(1+4\tilde k^2_1)^{3/2}}{128\tilde k^3_1t(3-4\tilde k^2_1)^{1/2}}}\\
{}\quad -\frac{1}{\pi}\int_{k_1}^{\frac{1}{k_1}}\ln\left(\frac{1+|r(s)|^2}{1+|r(k_1)|^2}\right)\frac{-2k_1}{s^2-k^2_1}ds-\frac{1}{\pi}\int_{-k_2}^{k_2}\ln\left(\frac{1+|r(s)|^2}{1+|r(k_2)|^2}\right)\frac{ds}{s-\frac{1}{k_1}}\\
{}\quad-\frac{1}{\pi}\int_{\frac{1}{k_2}}^{+\infty}\ln(\frac{s+\frac{1}{k_1}}{s-\frac{1}{k_1}})d\ln(1+|r(s)|^2)+2\gam_2\ln\left(\frac{(1-k_1k_2)(k_1+k_2)}{(1+k_1k_2)(k_1-k_2)}\right)\\
{}\quad -2\tilde k_1(x-\tilde c_+-\frac{2t}{1+4\tilde k^2_1}),\\

\phi_2=\frac{\pi}{4}-{\mbox arg}(r(k_2))+{\mbox arg}(\Gam(i\gam_2))-2\gam_2\ln{\frac{(1+4\tilde k^2_2)^{3/2}}{128\tilde k^3_2t(3-4\tilde k^2_2)^{1/2}}}\\
{}\quad -\frac{1}{\pi}\int_{k_1}^{\frac{1}{k_1}}\ln\left(\frac{1+|r(s)|^2}{1+|r(k_2)|^2}\right)\frac{-2k_2}{s^2-k^2_2}ds-\frac{1}{\pi}\int_{-k_2}^{k_2}\ln\left(\frac{1+|r(s)|^2}{1+|r(k_2)|^2}\right)\frac{ds}{s-\frac{1}{k_2}}\\
{}\quad-\frac{1}{\pi}\int_{\frac{1}{k_2}}^{+\infty}\ln(\frac{s+\frac{1}{k_2}}{\frac{1}{k_2}-s})d\ln(1+|r(s)|^2)-2\gam_1\ln\left(\frac{(1-k_1k_2)(k_1+k_2)}{(1+k_1k_2)(k_1-k_2)}\right)\\
{}\quad -2\tilde k_2(x-\tilde c_+-\frac{2t}{1+4\tilde k^2_2}),\\

\tilde c_+=-\frac{2}{\pi}\int_{k_1}^{\frac{1}{k_1}}\frac{\ln(1+|r(s)|^2)}{1+s^2}ds-\frac{2}{\pi}\int_{\frac{1}{k_2}}^{+\infty}\frac{\ln(1+|r(s)|^2)}{1+s^2}ds
-\frac{1}{\pi}\int_{-k_2}^{k_2}\frac{\ln(1+|r(s)|^2)}{1+s^2}ds.
\ea
\ee          
  
  \item Case 4: $\xi<-\frac{1}{4}$. In this case, the solution $u(x,t)$ of the mCH equation is decaying fast as $t\rightarrow \infty$.
\end{itemize}

\section{Soliton solutions}
To obtain the soliton solutions of the mCH equation, we need consider the zeros of the function $a(k)$. In the following, we assume that $a(k)$ has finite $N$  simple zeros which lie on the upper-plane of complex $k$. 

\subsection{Residue conditions}
\par
From the symmetry conditions (\ref{symrel}) and (\ref{s1krelsk}), we know that if $k_j\in D_1$ is a zero of $a(k)$, then so are $-\bar k_j, -\frac{1}{k_j}$ and $\frac{1}{\bar k_j}$.
\par
So, in view of the definition of $s(k)$ and $M(y,t,k)$, we have the following residue condition,
\be\label{rescon1}
\mbox{Res}_{k=k_j}[M(y,t,k)]_1=c_je^{Q(k_j)(y-\frac{2t}{\lam^2(k_j)})}[M(y,t,k_j)]_2,
\ee
with some constant $c_j$.
\par
Then, by the symmetry conditions (\ref{symrel}) and (\ref{s1krelsk}), again, we have,
\begin{subequations}\label{rescon2}
\be
\mbox{Res}_{k=-\bar k_j}[M(y,t,k)]_1=\bar c_je^{\ol{Q(k_j)(y-\frac{2t}{\lam^2(k_j)})}}[M(y,t,-\bar k_j)]_2,
\ee
\be
\mbox{Res}_{k=-\frac{1}{k_j}}[M(y,t,k)]_1=-\frac{c_j}{k^2_j}e^{Q(k_j)(y-\frac{2t}{\lam^2(k_j)})}[M(y,t,-\frac{1}{k_j})]_2,
\ee
\be
\mbox{Res}_{k=\frac{1}{\bar k_j}}[M(y,t,k)]_1=-\frac{\bar c_j}{\bar k^2_j}e^{\ol{Q(k_j)(y-\frac{2t}{\lam^2(k_j)})}}[M(y,t,\frac{1}{\bar k_j})]_2,
\ee
\end{subequations}

\subsection{Soliton solutions}

Again, by the symmetry conditions (\ref{symrel}), to obtain the soliton solutions of the mCH equation, we seek the solution $\tilde M(y,t,k)$ of the Riemann-Hilbert problem as follows,
\be
\tilde M(y,t,k)=\left(f(k),\quad f(-k)\right),
\ee
where $f(k)$ is a function which we omit the variables $(y,t)$.
\par
If we denote $Q(k_j)(y-\frac{2t}{\lam^2(k_j)})$ by $\phi_j$, then from (\ref{rescon1}) and (\ref{rescon2}),
\be\label{fkdef}
f(k)=1+\sum_{j=1}^{N}\left(\frac{c_je^{\phi_j}}{k-k_j}f(-k_j)+\frac{\bar c_j e^{\bar \phi_j}}{k+\bar k_j}f(\bar k_j)+\frac{-\frac{c_j}{k^2_j}e^{\phi_j}}{k+\frac{1}{k_j}}f(\frac{1}{k_j})+\frac{-\frac{\bar c_j}{\bar k^2_j}e^{\bar \phi_j}}{k-\frac{1}{\bar k_j}}f(-\frac{1}{\bar k_j})\right).
\ee
Evaluation at $-k_j,\bar k_j,\frac{1}{k_j}$ and $-\frac{1}{\bar k_j}$, respectively, yield
\be\label{fksys}
\left\{
\ba{l}
f(-k_j)=1+\sum_{j=1}^{N}\left(\frac{c_je^{\phi_j}}{-k_j-k_j}f(-k_j)+\frac{\bar c_j e^{\bar \phi_j}}{-k_j+\bar k_j}f(\bar k_j)+\frac{-\frac{c_j}{k^2_j}e^{\phi_j}}{-k_j+\frac{1}{k_j}}f(\frac{1}{k_j})+\frac{-\frac{\bar c_j}{\bar k^2_j}e^{\bar \phi_j}}{-k_j-\frac{1}{\bar k_j}}f(-\frac{1}{\bar k_j})\right)\\
f(\bar k_j)=1+\sum_{j=1}^{N}\left(\frac{c_je^{\phi_j}}{\bar k_j-k_j}f(-k_j)+\frac{\bar c_j e^{\bar \phi_j}}{\bar k_j+\bar k_j}f(\bar k_j)+\frac{-\frac{c_j}{k^2_j}e^{\phi_j}}{\bar k_j+\frac{1}{k_j}}f(\frac{1}{k_j})+\frac{-\frac{\bar c_j}{\bar k^2_j}e^{\bar \phi_j}}{\bar k_j-\frac{1}{\bar k_j}}f(-\frac{1}{\bar k_j})\right)\\
f(\frac{1}{k_j})=1+\sum_{j=1}^{N}\left(\frac{c_je^{\phi_j}}{\frac{1}{k_j}-k_j}f(-k_j)+\frac{\bar c_j e^{\bar \phi_j}}{\frac{1}{k_j}+\bar k_j}f(\bar k_j)+\frac{-\frac{c_j}{k^2_j}e^{\phi_j}}{\frac{1}{k_j}+\frac{1}{k_j}}f(\frac{1}{k_j})+\frac{-\frac{\bar c_j}{\bar k^2_j}e^{\bar \phi_j}}{\frac{1}{k_j}-\frac{1}{\bar k_j}}f(-\frac{1}{\bar k_j})\right)\\
f(-\frac{1}{\bar k_j})=1+\sum_{j=1}^{N}\left(\frac{c_je^{\phi_j}}{-\frac{1}{\bar k_j}-k_j}f(-k_j)+\frac{\bar c_j e^{\bar \phi_j}}{-\frac{1}{\bar k_j}+\bar k_j}f(\bar k_j)+\frac{-\frac{c_j}{k^2_j}e^{\phi_j}}{-\frac{1}{\bar k_j}+\frac{1}{k_j}}f(\frac{1}{k_j})+\frac{-\frac{\bar c_j}{\bar k^2_j}e^{\bar \phi_j}}{-\frac{1}{\bar k_j}-\frac{1}{\bar k_j}}f(-\frac{1}{\bar k_j})\right).
\ea
\right.
\ee
Solving this algebraic system for $f(-k_j),f(\bar k_j),f(\frac{1}{k_j})$ and $f(-\frac{1}{\bar k_j})$, $j=1,2,\dots, N$, and and substituting
the solution into (\ref{fkdef}) yields an explicit expression for $f(k)$. This solves the
Riemann-Hilbert problem for $\tilde M(y,t,k)$. Therefore, by (\ref{usol}), we can obtain an parametric expression for $u(x,t)$.
\par
In the following, we assume $N=1$ and derive an explicit formula for the one-soliton solution in two special cases. 
\subsubsection{One-soliton solution for $|k_1|=1$}
In this case, there are two zeros of $a(k)$, i.e., one is $k_1$, another is $-\bar k_1$. It yields that the algebraic system (\ref{fksys}) reduces to the following two equations
\be\label{fksyscase1}
\left\{
\ba{l}
f(-k_1)=1+\frac{c_1e^{\phi_1}}{-k_1-k_1}f(-k_1)+\frac{\bar c_1 e^{\bar \phi_1}}{-k_1+\bar k_1}f(\bar k_1)\\
f(\bar k_1)=1+\frac{c_1e^{\phi_1}}{\bar k_1-k_1}f(-k_1)+\frac{\bar c_1 e^{\bar \phi_1}}{\bar k_1+\bar k_1}f(\bar k_1).
\ea
\right.
\ee
If denote $k_1=e^{i\alpha_1}=a_1+ib_1$, $c_1=|c_1|e^{iC_1}$, $\phi_1=\psi_1+i\psi_2$, where $\alpha_1,a_1,b_1,C_1,\psi_1,\psi_2$ are some real constants, then solving the system (\ref{fksyscase1}) for $f(-k_1)$ and $f(\bar k_1)$, we have,
\begin{subequations}\label{fkcase1}
\be
f(-k_1)=\frac{1-\frac{a_1}{2ib_1}|c_1|e^{\psi_1}e^{-i(C_1+\psi_2-\alpha_1)}}{1+\frac{a^2_1}{4b^2_1}|c_1|^2e^{2\psi_1}+i|c_1|e^{\psi_1}\sin(C_1+\psi_2-\alpha_1)},
\ee
\be
f(\bar k_1)=\frac{1-\frac{a_1}{2ib_1}|c_1|e^{\psi_1}e^{i(C_1+\psi_2-\alpha_1)}}{1+\frac{a^2_1}{4b^2_1}|c_1|^2e^{2\psi_1}+i|c_1|e^{\psi_1}\sin(C_1+\psi_2-\alpha_1)},
\ee
\end{subequations}
If we choose $|c_1|=\frac{2b_1}{a_1}\mbox{sgn}(a_1)$ and $C_1+\psi_2-\alpha_1=\frac{\pi}{2}$, then (\ref{fkcase1}) can be written as
\be
f(-k_1)=H_1+H_2,\quad f(\bar k_1)=H_1-H_2,
\ee
where
\be
H_1=\frac{1}{1+e^{2\psi_1}+i\frac{2b_1}{a_1}\mbox{sgn}(a_1)e^{\psi_1}},\quad H_2=\frac{\mbox{sgn}(a_1)e^{\psi_1}}{1+e^{2\psi_1}+i\frac{2b_1}{a_1}\mbox{sgn}(a_1)e^{\psi_1}}.
\ee
Hence, by (\ref{usol}), we get
\be\label{usolcase1}
u(y,t)=-\frac{4b_1}{a^3_1}\mbox{sgn}(a_1)\frac{\cosh(\psi_1)}{\cosh(2\psi_1)+\frac{1+b^2_1}{1-b^2_1}},
\ee
\be
c_+=\ln\left(\frac{1+\frac{1-b}{1+b}e^{2\psi_1}}{1+\frac{1+b}{1-b}e^{2\psi_1}}\right), y=x-c_+,
\ee
\be
m(y,t)=-2a_1b_1\mbox{sgn}(a_1)\frac{\mbox{sech}(\psi_1)}{a^2_1-2b^2_1\mbox{sech}^2(\psi_1)}.
\ee

\subsubsection{One-soliton solution for $k_1\in i\R$}
In this case,  there are two zeros of $a(k)$, i.e., one is $k_1$, another is $-\frac{1}{k_1}$. It yields that the algebraic system (\ref{fksys}) reduces to the following two equations
\be\label{fksyscase2}
\left\{
\ba{l}
f(-k_1)=1+\frac{c_1e^{\phi_1}}{-k_1-k_1}f(-k_1)+\frac{-\frac{c_1}{k^2_1} e^{\phi_1}}{-k_1+\frac{1}{k_1}}f(\frac{1}{k_1})\\
f(\frac{1}{k_1})=1+\frac{c_1e^{\phi_1}}{\frac{1}{k_1}-k_1}f(-k_1)+\frac{-\frac{c_1}{k^2_1} e^{\phi_1}}{\frac{1}{k_1}+\frac{1}{k_1}}f(\frac{1}{k_1}).
\ea
\right.
\ee
If we denote $k_1=ib$, where $b<1$, then solving the system (\ref{fksyscase2}) for $f(-k_1)$ and $f(\frac{1}{k_1})$, we have,
\begin{subequations}\label{fkcase2}
\be
f(-k_1)=\frac{1-\frac{\frac{1}{b}-b}{2ib(b+\frac{1}{b})}c_1e^{\phi_1}}{1-\frac{(b-\frac{1}{b})^2}{4b^2(b+\frac{1}{b})^2}c^2_1e^{2\phi_1}+\frac{c_1e^{\phi_1}}{ib}}
\ee
\be
f(\frac{1}{k_1})=\frac{1+\frac{\frac{1}{b}-b}{2ib(b+\frac{1}{b})}c_1e^{\phi_1}}{1-\frac{(b-\frac{1}{b})^2}{4b^2(b+\frac{1}{b})^2}c^2_1e^{2\phi_1}+\frac{c_1e^{\phi_1}}{ib}}
\ee
\end{subequations}
If we choose $c_1=\frac{2b(1+b^2)}{1-b^2}$, then (\ref{fkcase2}) and (\ref{usol}) yield
\be\label{usolcase2}
u(y,t)=-\frac{16b^2(1+b^2)}{(b^2-1)^3}\frac{\sinh(\phi_1)}{\cosh(2\phi_1)+\frac{(1+b^2)^2+4b^2}{(1-b^2)^2}},
\ee
\be
c_+=\ln\left(\frac{1+\left(\frac{1-b}{1+b}\right)^2e^{2\phi_1}}{1+\left(\frac{1+b}{1-b}\right)^2e^{2\phi_1}}\right), y=x-c_+,
\ee
\be
m(y,t)=4\frac{1+b^2}{1-b^2}\frac{\sinh(\phi_1)}{\cosh(2\phi_1)-\frac{3(1+b^4)+2b^2}{(1-b^2)^2}}.
\ee
\begin{remark}
The one-soliton solution obtained in (\ref{usolcase1}) and (\ref{usolcase2}) are the same as the equations (3.2a) and (3.9a) of \cite{matsuno2014} which were obtained by bilinear form method.
\end{remark}

\bigskip
{\bf Acknowledgements}
This work of Xu was supported by National Science Foundation of China under project NO.11971313 and Shanghai natural science foundation under project NO.19ZR1434500. Fan was support by grants from the National Science Foundation of China under project NO. 11671095, 51879045. We would like to thank Jonatan Lenells at KTH, Zhijun Qiao at UTGA, Ling Huang at KTH and Fudan University, Caiqing Song at USST, for their useful discussions and suggestions. We also thank the organizers of the forth China-Japan Joint Workshop on Integrable Systems held on August 19-22, 2019 for giving us an chance to report this work.

\end{document}